\documentclass[letterpaper, 11pt]{article}
\usepackage{appendix}

\usepackage{titlesec}
\titlespacing{\section}{0pt}{1.5ex}{1ex}
\titlespacing{\subsection}{0pt}{1.3ex}{0.5ex}
\titlespacing{\subsubsection}{0pt}{1.3ex}{0.5ex}
\titlespacing{\paragraph}{0pt}{1.3ex}{2ex}

\newcommand{\comment}[1]{}

\usepackage{graphicx}
\usepackage{subfigure}
\usepackage{amsmath}
\usepackage{mdwlist}
\usepackage{pxfonts}

\usepackage{algorithm}

\usepackage{xspace}
\newcommand{\BSB}{{\tt Broadcast\_Default}\xspace}

\newcommand{\sg}{{\mathcal G}}
\newcommand{\sv}{{\mathcal V}}
\newcommand{\se}{{\mathcal E}}

\newcommand{\sa}{{\mathcal A}}

\newcommand{\mybox}[1]{\vspace{5pt}\centerline{\framebox{\parbox[c]{0.97\textwidth}{#1}}}\vspace{5pt}}

\newenvironment{proof}{\paragraph{\bf Proof:}}{\hspace*{\fill}\(\Box\)}
\newenvironment{proofSketch}{\paragraph{\bf Proof Sketch:}}{\hspace*{\fill}\(\Box\)}
\newtheorem{theorem}{Theorem}

\newtheorem{lemma}{Lemma}

\def\noflash#1{\setbox0=\hbox{#1}\hbox to 1\wd0{\hfill}}

\begin{document}
\title{Network-Aware Byzantine Broadcast\\
in Point-to-Point Networks
 using Local Linear Coding
 \footnote{\normalsize This research is supported
in part by Army Research Office grant W-911-NF-0710287 and National
Science Foundation award 1059540. Any opinions, findings, and conclusions or recommendations expressed here are those of the authors and do not
necessarily reflect the views of the funding agencies or the U.S. government.}
}

\author{Guanfeng Liang and Nitin Vaidya\\ \normalsize Department of Electrical and Computer Engineering, and\\ \normalsize Coordinated Science Laboratory\\ \normalsize University of Illinois at Urbana-Champaign\\ \normalsize gliang2@illinois.edu, nhv@illinois.edu}

\date{}
\maketitle



\begin{abstract}
{\normalsize
The goal of Byzantine Broadcast (BB) is to allow a set
of fault-free nodes to agree on information that a source node wants to
broadcast to them, in the presence of Byzantine faulty nodes.
We consider design of efficient
algorithms for BB in {\em synchronous} point-to-point networks, where the rate of transmission over each communication link is limited by its "link capacity". The throughput of a particular BB algorithm is defined as the average number of bits that can be reliably broadcast to all fault-free nodes per unit time using the algorithm without violating the link capacity constraints.
The {\em capacity} of BB in a given network is then defined
as the supremum of all achievable BB throughputs in the given network, over all possible BB algorithms. \\
 
We develop NAB -- a Network-Aware Byzantine broadcast algorithm  -- for arbitrary point-to-point networks consisting of $n$
nodes, wherein the
number of faulty nodes is at most $f$, $f<n/3$, and the network connectivity is at least $2f+1$.
We also prove an upper bound on the capacity of Byzantine broadcast, and conclude that NAB can achieve throughput at least 1/3 of the capacity. When the network satisfies an additional condition, NAB can achieve throughput at least 1/2 of the capacity. \\ 

To the best of our knowledge, NAB is the first algorithm that can achieve a
constant fraction of capacity of Byzantine Broadcast (BB) in arbitrary point-to-point networks. \\
}
\end{abstract}

\thispagestyle{empty}

\newpage

\setcounter{page}{1}

\section{Introduction}
\label{sec:intro}

The problem of Byzantine Broadcast (BB) -- also known as the Byzantine
Generals problem \cite{psl82} --  was introduced by Pease, Shostak and
Lamport in their 1980 paper \cite{psl80}.
Since the first paper on this topic, Byzantine Broadcast has been the subject
of intense research activity, due to its many potential practical
applications, 
including replicated fault-tolerant state machines \cite{PBFT},
and fault-tolerant distributed file storage \cite{dfs}.
Informally, Byzantine Broadcast (BB) can be described as follows (we will define the problem more formally later). There is a source node that needs to
broadcast a message (also called its {\em input}) to all the other nodes such
that even if some of the nodes are {\em Byzantine faulty}, all the fault-free
nodes will still be able to agree on an identical message; the agreed message 
is identical to the source's input if the source is fault-free.

We consider the problem of maximizing the {\em throughput} of  
Byzantine Broadcast (BB) in {\em synchronous} networks of point-to-point links,
wherein each directed communication link is subject to a "capacity" constraint.
Informally speaking, {\em throughput} of BB is the number of bits of
Byzantine Broadcast that can be achieved per unit time (on average), under the worst-case
behavior by the faulty nodes. 
Despite the large body of work on BB \cite{Lynch:EasyImpossibilityProof,opt_bit_Welch92,bit_optimal_89,King:PODC2010,Hirt:linear,ashish:2009},
performance of BB in {\em arbitrary}\,
point-to-point network has not been investigated previously.
When capacities of the different links are not identical,
previously proposed algorithms can perform poorly.
In fact, one can easily construct example networks in which
previously proposed algorithms achieve throughput that is arbitrarily
worse than the optimal throughput.

\paragraph{Our Prior Work:}
In our prior work, we have considered the problem of optimizing
throughput of Byzantine Broadcast in 4-node networks
\cite{infocom-ByzantineBroadcast}. By comparing with an upper bound
on the capacity of BB in 4-node networks, we showed that our 4-node
algorithm is optimal. Unfortunately, the 4-node algorithm does not
yield very useful insights on design of good algorithms for larger networks.
This paper presents an algorithm that uses a different approach than
that in \cite{infocom-ByzantineBroadcast}, and also develops a different
upper bound on capacity that is helpful in our analysis of the new
algorithm. In other related work, we explored design of efficient Byzantine
consensus algorithms when {\em total communication cost}
is the metric (which is oblivious of link capacities) \cite{podc2011_consensus}. 

%
\paragraph{Main contributions:} This  paper studies throughput and capacity of Byzantine broadcast in {\em arbitrary} point-to-point networks.

\begin{enumerate}
\item We develop a Network-Aware Byzantine (NAB) broadcast algorithm for {\em arbitrary} point-to-point networks wherein each directed communication link is subject to a capacity constraint. 
The proposed NAB algorithm is ``network-aware'' in the sense that its
design takes the link capacities into account.

\item We derive an upper bound on the capacity of BB in arbitrary point-to-point networks.

\item We show that NAB can achieve throughput at least 1/3 of the capacity in arbitrary point-to-point networks. When the network satisfies an additional condition, NAB can achieve throughput at least 1/2 of the capacity.



\end{enumerate}


\noindent
We consider a {\em synchronous} system consisting of $n$ nodes,
named $1,2,\cdots,n$, with
one node designated as the {\em sender} or {\em source} node.
In particular, we will assume that node 1 is
the source node. Source node 1 is given an {\em input} value $x$ containing $L$ bits,
and the goal here is for the source to broadcast its input to all the other
nodes. 
The following conditions must be satisfied when the input value at the source
node is $x$:
\begin{itemize}
\item \textbf{Termination:} Every fault-free node $i$ must eventually decide on an {\em output} value of $L$ bits; let us denote the output value of fault-free node $i$ as $y_i$.
\item
\textbf{Agreement:} All fault-free nodes must agree on an identical output value, i.e., there exists $y$ such that $y_i=y$ for each fault-free node $i$.
\item
\textbf{Validity:} If the source node is fault-free, then the agreed value
must be identical to the input value of the source, i.e., $y=x$.
\end{itemize}
\paragraph{Failure Model:} The faulty nodes are controlled by an adversary that has a complete knowledge of the network topology, the algorithm, and the information the source is trying to send. No secret is hidden from the adversary.
The adversary can take over up to  $f$
nodes at any point during execution of the algorithm, where
$f < n/3$. These nodes are said to be {\em faulty}. The faulty nodes can engage in any kind of
deviations from the algorithm, including sending incorrect
or inconsistent messages to the neighbors.

We assume that the set of faulty nodes remains fixed across different instances
of execution of the BB algorithm.
This assumption captures the conditions in practical replicated server
systems. In such a system, the replicas may use Byzantine Broadcast
to agree on requests to be processed. The set of faulty (or compromised)
replicas that may adversely affect the agreement on each request does
{\em not} change arbitrarily. We model this by assuming that the set of
faulty nodes remains fixed over time.

When a faulty node fails to send a message to a neighbor as required by
the algorithm, we assume that the recipient node interprets the
missing message as being some default value.

\paragraph{Network Model:}
We assume a synchronous point-to-point network modeled as a directed simple
graph $\sg(\sv,\se)$, where the set of vertices $\sv=\{1,2,\cdots,n\}$
represents the nodes in the point-to-point network,
and the set of edges $\se$ represents the links in the network. 
 The {\em capacity} of an edge $e\in \se$ is denoted as $z_e$.
With a slight abuse of terminology, we will use the terms {\em edge} and {\em link} interchangeably, and use the terms {\em vertex} and {\em node} interchangeably.
We assume that $n\geq 3f+1$ and that the network connectivity is at least $2f+1$
(these two conditions are necessary for the existence of a correct BB algorithm \cite{Lynch:EasyImpossibilityProof}).

In the given network,
links may not exist between all node pairs. Each directed link is associated with a {\em fixed} link capacity, which specifies the maximum amount of information that can be transmitted on that link per unit time. Specifically, over a directed edge $e=(i,j)$ with capacity $z_e$ bits/unit time,
we assume that up to $z_e\tau$ bits can be reliably sent
from node $i$ to node $j$ over time duration $\tau$ (for
any non-negative $\tau$). 
This is a deterministic model of {\em capacity}\, that has been commonly
used in other work \cite{Li03linearnetwork, Cai06networkerror, Ho04byzantinemodification, Jaggi_infocom07}. 
All link capacities are assumed to be positive integers. Rational link capacities can be turned into integers by choosing a suitable time unit. Irrational link capacities can be approximated by integers with arbitrary accuracy by choosing a
suitably long time unit. Propagation delays on the links are assumed to be zero (relaxing this
assumption does not impact the correctness of results shown for large input sizes).
We also assume that each node correctly knows the identity of the nodes at the other end of its links.


\subsection*{Throughput and Capacity of Byzantine Broadcast}
When defining the throughput of a given BB algorithm in a given network,
we consider $Q\ge 1$ independent instances of BB. The source node is given an
$L$-bit input for each of these $Q$ instances, and the {\em validity}
and {\em agreement} properties need to be satisfied for each instance
{\em separately} (i.e., independent of the outcome for the other instances).

For any BB algorithm $\sa$, denote $t(\sg,L,Q,\sa)$ as the duration of time
required, in the worst case, to complete $Q$ instances of $L$-bit Byzantine Broadcast, without violating the capacity constraints of the links in $\sg$.
Throughput of algorithm $\sa$ in network $\sg$ for $L$-bit inputs is then defined as
\[ T(\sg,L,\sa) = \lim_{Q\rightarrow \infty}\,\frac{LQ}{t(\sg,L,Q,\sa)}\]
We then define capacity $C_{BB}$ as follows.

\mybox{
\vspace{5pt}
Capacity $C_{BB}$ of Byzantine Broadcast in network $\sg$ is defined as the
supremum over the throughput of all algorithms $\sa$ that solve the BB problem
and all values of $L$.
That is,
\begin{equation}
C_{BB}(\sg) ~ = ~ \sup_{\sa, L} \, T(\sg,L,\sa).
\end{equation}
\vspace{-10pt}
}


\section{Algorithm Overview}
\label{sec:overview}

Each instance of our NAB algorithm performs Byzantine broadcast of
an $L$-bit value. We assume that the NAB algorithm is used repeatedly,
and during all these repeated executions, the cumulative number of faulty
nodes is upper bounded by $f$. Due to this assumption,
the algorithm can perform well by amortizing the cost of fault tolerance
over a large number of executions. Larger values of $L$ also result in
better performance for the algorithm. The algorithm is intended to be used for sufficiently large $L$, to be elaborated later. 

The $k$-th instance of NAB executes on a network corresponding to graph
$\sg_k(\sv_k,\se_k)$, defined as follows:
\begin{itemize}
\item For the first instance, $k=1$, and $\sg_1=\sg$. Thus,
$\sv_1=\sv$ and $\se_1=\se$.

\item The $k$-th instance of NAB occurs on graph $\sg_k$ in the
following sense: (i) all the fault-free nodes know the node and
edge sets $\sv_k$ and $\se_k$, (ii)
only the nodes corresponding to the vertices in $\sv_k$ need to participate
in the $k$-th instance of BB, and (iii) only the links corresponding to the
edges in $\se_k$ are used for communication in the $k$-th instance
of NAB (communication received on other links can be ignored).

During the $k$-th instance of NAB using graph $\sg_k$, if misbehavior by some
faulty node(s) is detected, then, as described later, additional information is
gleaned about the potential identity of the faulty node(s). In this case,
$\sg_{k+1}$ is obtained by removing from $\sg_k$ appropriately chosen edges
and possibly some vertices (as described later).

On the other hand, if during the $k$-th instance, no misbehavior is detected,
then $\sg_{k+1}=\sg_k$.

\end{itemize}
The $k$-th instance of NAB algorithm consists of three phases, as described
next. The main contributions of this paper are (i) the algorithm used in Phase
2 below, and (ii) a performance analysis of NAB.

If graph $\sg_k$ does {\bf not} contain the source node 1, then (as will
be clearer later) by the start of the $k$-th instance of NAB, all the fault-free
nodes already know that the source
node is surely faulty;
in this case, the fault-free nodes can agree on a default value for the
output, and terminate the algorithm. Hereafter, we will assume that
the source node 1 is in $\sg_k$.

\newpage

\thispagestyle{empty}

\begin{figure}[!ht]
  \centering
\subfigure[]{\label{fig:original}\includegraphics[width=0.3\textwidth]{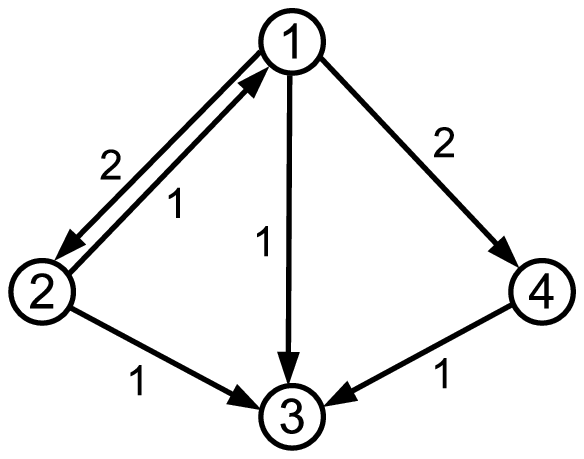}}  
\hspace{0.05\textwidth}  
  \subfigure[]{ \label{fig:new}\includegraphics[width=0.3\textwidth]{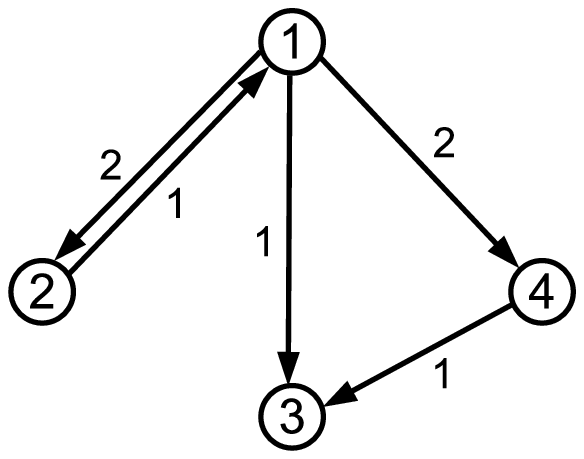}}
\hspace{0.05\textwidth}  
  \caption{Example graphs}
\label{fig:example}
\end{figure}

~

~

\begin{figure}[!ht]
  \centering
  \subfigure[Directed graph $G$]{\label{fig:directed}\includegraphics[width=0.3\textwidth]{directed}}  
\hspace{0.05\textwidth}  
  \subfigure[Undirected graph $\overline{G}$]{\label{fig:undirected}\includegraphics[width=0.3\textwidth]{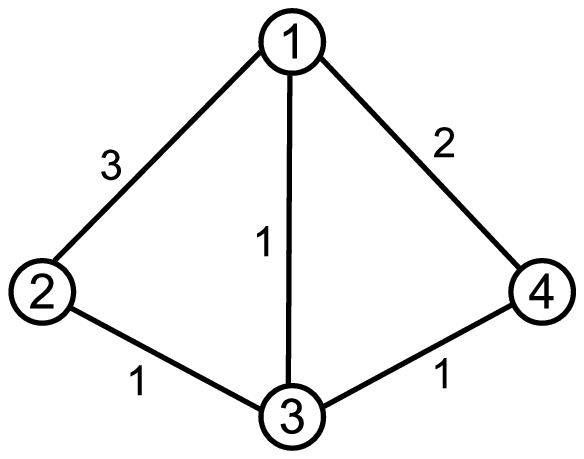}}
\\ 
\subfigure[Two unit-capacity spanning trees in the directed graph. Every directed edge has capacity 1]{ \label{fig:directed-trees}\includegraphics[width=0.3\textwidth]{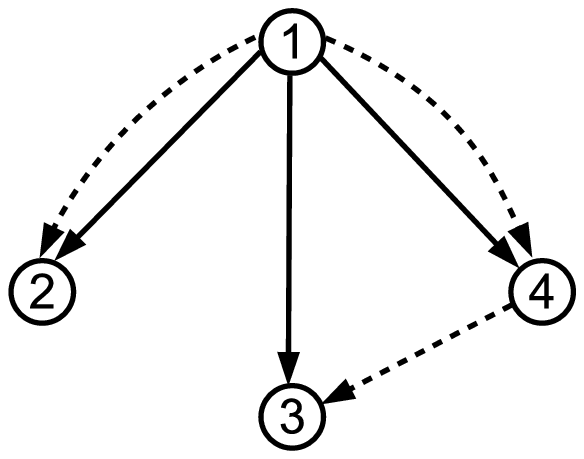}}
 \hspace{0.05\textwidth}
  \subfigure[A spanning tree in the undirected graph shown in dotted edges]{\label{fig:undirected-tree}\includegraphics[width=0.3\textwidth]{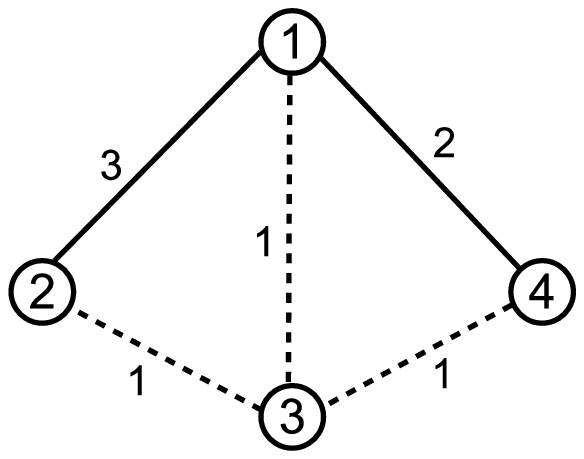}}
 \caption{Different graph representations of a network. Numbers next to the edges indicate link capacities.}
\label{fig:graph:representations}
\end{figure}


%

\newpage

\addtocounter{page}{-1}

\subsection*{Phase 1: Unreliable Broadcast}

In Phase 1, source node 1 broadcasts $L$ bits to all the other nodes in $\sg_k$.
This phase makes no effort to detect or tolerate misbehavior by faulty
nodes.
As elaborated in
Appendix~\ref{app:unreliable_broadcast}, unreliable broadcast
can be performed using a set of spanning trees embedded in graph $\sg_k$.
Now let us analyze the time required to perform unreliable broadcast in Phase 1.

$MINCUT(\sg_k,1,j)$ denotes the minimum cut in the directed graph $\sg_k$
from source node 1 to node $j$.
Let us define
 $$\gamma_k = \min_{j\in \sv_k}~MINCUT(\sg_k,1,j).$$
$MINCUT(\sg_k,1,j)$ is equal
to the maximum flow rate possible from node 1 to node $j\in \sv_k$. It is well-known \cite{maxflowmincut} that 
$\gamma_k$ is 
the maximum rate achievable for unreliable broadcast from node 1 to {\em all} the other nodes in $\sv_k$, under the
capacity constraints on the links in $\se_k$.
Thus, the least amount of time in which
$L$ bits can be broadcast by node 1 in graph $\sg_k$ is given by\footnote{\normalsize To
simplify the analysis, we ignore propagation delays. Analogous  results on throughput and capacity can
be obtained in the presence of propagation delays as well.}
\begin{eqnarray}
& L~/~\gamma_k & \label{e_phase1_time}
\end{eqnarray}
Clearly, $\gamma_k$ depends on the capacities of the links in $\sg_k$.
For example, if $\sg_k$ were the directed graph in Figure \ref{fig:original}, 
then $MINCUT(\sg_k,1,2) = MINCUT(\sg_k,1,4)=2$, $MINCUT(\sg_k,1,3)=3$, and
hence $\gamma_k = 2$.

At the end of the broadcast operation in Phase 1, each node should have received $L$ bits.
At the end of Phase 1 of the $k$-th instance of NAB,
one of the following four outcomes will occur:
\begin{list}{}{}
\item[(i)]
The source node 1 is fault-free, and all the fault-free nodes correctly receive 
the source node's $L$-bit
input for the $k$-th instance of NAB, or
\item[(ii)] The source node 1 is fault-free, but
some of the fault-free nodes receive incorrect $L$-bit values due to misbehavior
by some faulty node(s), or
\item[(iii)] The source node 1 is faulty, but all the
fault-free nodes still receive an identical $L$-bit value in Phase 1, or
\item[(iv)] The source node is faulty, and all the fault-free nodes
do not receive an identical $L$-bit value in Phase 1.
\end{list}
The values received by the fault-free nodes in cases (i) and (iii) satisfy the
{\em agreement} and {\em validity} conditions,
 whereas in cases (ii) and (iv) at least one of the two conditions
is violated.

\subsection*{Phase 2: Failure Detection}

\comment{=================================================
Outcome from Phase 2 will be one of the following:
\begin{itemize}
\item (F1) The fault-free nodes agree on the $L$-bit values they received in Phase 1,
	and the agreed values satisfy the agreement and validity conditions.
	In this case, the $k$-th instance of NAB ends at the ends of Phase 2 
	({\bf without performing Phase 3}).
\item (F2) All the fault-free nodes become aware of misbehavior by some faulty node(s).
This misbehavior may have occurred during Phase 1 or Phase 2.
In this case, {\bf Phase 3 is performed}.
Note that even though the fault-free nodes may become aware of the misbehavior,
they may not necessarily know the identity of the faulty nodes.
\end{itemize}
The algorithm for Phase 2, presented in Section~\ref{s_phase_2},
and its analysis, are our main contributions.
====================================================}

Phase 2 performs the following two operations. As stipulated in
the fault model, a faulty node may not follow the algorithm
specification correctly.
\begin{itemize}
\item (Step 2.1) {\em Equality check:}
Using an {\em Equality Check} algorithm,
the nodes in $\sv_k$ perform a comparison of the $L$-bit value
	they received in Phase 1, to determine
if all the nodes received an identical value. The source node 1 also 
	participates in this comparison operation (treating its input
	as the value ``received from'' itself). 

Section \ref{sec:eq} presents the {\em Equality Check} algorithm,
which is designed to {\bf guarantee} that if the values received
by the fault-free nodes in Phase 1 are {\bf not identical}, then at
least one fault-free node will detect the {\bf mismatch}. 

\item (Step 2.2) {\em Agreeing on the outcome of equality check:}
Using a previously proposed Byzantine broadcast algorithm, such as
\cite{psl80}, each node performs Byzantine broadcast of
 a 1-bit {\em flag} to other nodes in $\sg_k$ indicating whether
it detected a mismatch during {\em Equality Check}.
\end{itemize}

If any node broadcasts in step 2.2 that it has detected a mismatch,
then subsequently {\bf Phase 3 is performed}.
On the other hand, if no node announces a mismatch in step 2.2 above,
then {\bf Phase 3 is not performed}; in this case,
each fault-free node agrees on the value it received in Phase 1,
and the $k$-th instance of {\bf NAB is completed}.

We will later prove that, when Phase 3 is {\bf not} performed,
the values agreed above by the fault-free nodes satisfy the
{\em validity} and {\em agreement} conditions for the $k$-th
instance of NAB.
On the other hand, when Phase 3 is performed during the $k$-th instance of NAB, as noted below,
Phase 3 results in correct outcome for the $k$-th instance.

When Phase 3 is performed, Phase 3 determines $\sg_{k+1}$.
Otherwise, $\sg_{k+1}=\sg_k$.

\subsection*{Phase 3: Dispute Control}

Phase 3 employs a {\em dispute control}\, mechanism that has also been
used in prior work \cite{Hirt06DisputeControl,podc2011_consensus}.
Appendix~\ref{app:diagnosis} provides the details of
the dispute control algorithm used in Phase 3.
Here we summarize the outcomes of this phase -- this summary should
suffice for understanding the main contributions of this paper.

The dispute control in Phase 3 has very high overhead, due to the large amount of data that needs to be transmitted. From the above discussion of Phase 2,
it follows that Phase 3 is performed {\bf only if}
at least one faulty node misbehaves during Phases 1 or 2.
The outcomes from Phase 3 performed during the $k$-th
instance of NAB are as follows.
\begin{itemize}
\item Phase 3 results in correct Byzantine broadcast for the
$k$-th instance of NAB.
This is obtained as a byproduct of the Dispute Control mechanism.

\item By the end of Phase 3, either
one of the nodes in $\sv_k$ is correctly identified
as faulty, or/and at least one pair of nodes in $\sv_k$, say nodes $a,b$,
is identified as being ``in dispute'' with each other. When a node pair $a,b$
is found {\em in dispute}, it is guaranteed that
(i) {\em at least}\, one of these two nodes is faulty,
and (ii) at least one of the directed edges $(a,b)$ and $(b,a)$ is in $\se_k$.
Note that
the dispute control phase {\bf never} finds two fault-free nodes
in dispute with each other.

\item Phase 3 in the $k$-th instance computes graph $\sg_{k+1}$.
In particular, any nodes that can be inferred as being faulty based
on their behavior so far are excluded from $\sv_{k+1}$; links
attached to such nodes are excluded from $\se_{k+1}$. In Appendix \ref{app:diagnosis}
we elaborate on how the faulty nodes are identified. Then,
for each node pair in $\sv_{k+1}$, if that node
pair has been found in dispute at least in one instance of NAB so far,
the links between the node pair are excluded from $\se_{k+1}$.
Phase 3 ensures that
all the fault-free nodes compute an identical graph $\sg_{k+1}=(\sv_{k+1},\se_{k+1})$
to be used during the next instance of NAB.

%
\end{itemize}
Consider two special cases for the $k$-th instance of NAB:
\begin{itemize}
\item If graph $\sg_k$ does not contain the source node 1, it implies that all the
fault-free nodes are aware that node 1 is faulty. In this case, they can safely agree
on a default value as the outcome for the $k$-th instance of NAB.
\item 
Similarly, if the source node is in $\sg_k$ but at least $f$ other nodes
are excluded from $\sg_k$, that implies that the remaining nodes in $\sg_k$
are all fault-free; in this case, algorithm NAB can be reduced to just Phase 1.

\end{itemize}

Observe that during each execution of Phase 3, either a new pair of nodes
{\em in dispute}\, is identified, or a new node is identified as faulty.
Once a node is found to be in dispute
with $f+1$ distinct nodes, it can be identified as faulty, and excluded from the algorithm's
execution.
Therefore, Dispute Control 
needs to be performed at most $f(f+1)$ times over repeated executions of NAB. Thus, even though each dispute control phase is expensive, the bounded number ensures that the amortized cost over a large number of instances of NAB is small, as reflected in the performance analysis of NAB (in Section \ref{sec:throughput} and Appendix \ref{app:BB:throughput}).

\comment{================================================
\section{Phase 2: Failure Detection}
\label{s_phase_2}

Phase 2 performs the following two operations. As stipulated in
the fault model, a faulty node may not follow the algorithm
specification correctly.
\begin{itemize}
\item (Step 2.1) {\em Equality check:}
The nodes in $\sv_k$ perform a comparison of the $L$-bit value
	they received in Phase 1, to determine
if all the nodes received an identical value. The source node 1 also 
	participates in this comparison operation (treating its input
	as the value ``received from'' itself). 
Section \ref{sec:eq} presents the algorithm for {\em Equality Check}.

The algorithm used for comparison {\bf guarantees} that if the values received
by the fault-free nodes in Phase 1 are {\bf not identical}, at least one fault-free
node will detect the {\bf mismatch}. 

\item (Step 2.2) {\em Agreeing on the outcome of equality check:}
Using a previously proposed Byzantine broadcast algorithm, such as
\cite{psl80}, each node performs Byzantine broadcast of
 a 1-bit {\em flag} to other nodes in $\sg_k$ indicating whether
it detected a mismatch during {\em Equality Check}.
If any node thus announces that it detected a mismatch,
then subsequently {\bf Phase 3 is performed}; otherwise,
each fault-free node agrees on the value it received in Phase 1, and the
$k$-th instance of NAB is completed.
\end{itemize}
=============================================================}

\vspace{-5pt}
\section{Equality Check Algorithm with Parameter $\rho_k$}
\label{sec:eq}

We now present the {\em Equality Check} algorithm used in Phase 2, which has an integer parameter
$\rho_k$ for the $k$-th instance of NAB. Later in this section, we will elaborate
on the choice of $\rho_k$, which is dependent on capacities of the links in $\sg_k$.

Let us denote by $x_i$ the $L$-bit value received by fault-free node $i\in \sv_k$ in
Phase 1 of the $k$-th instance. For simplicity, we do not include index
$k$ in the notation $x_i$. To simplify the presentation, let us assume that
$L/\rho_k$ is an integer. Thus we can represent the $L$-bit value $x_i$ as 
$\rho_k$ symbols from Galois Field $GF(2^{L/\rho_k})$.
In particular, we represent $x_i$ as a vector $\bf X_i$,
\vspace{-5pt}$$ {\bf X_i} = [X_i(1), \, X_i(2),\cdots,X_i(\rho_k)] \vspace{-5pt}$$
where each symbol $X_i(j)\in GF(2^{L/\rho_k})$ can be represented using $L/\rho_k$ bits.
As discussed earlier, for convenience, we assume that all the link capacities are integers when using a suitable time unit.

\begin{algorithm}[h]
\caption{Equality Check in $\sg_k$ with parameter $\rho_k$}
\label{alg:MEQ:new}
Each node $i\in\sv_k$ should performs these steps:
\begin{enumerate}
\item On each outgoing link $e=(i,j)\in\se_k$ whose capacity is $z_e$, node $i$ transmits
$z_e$ linear combinations of the $\rho_k$ symbols in vector
$\bf X_i$, with the weights for the linear combinations being chosen from $GF(2^{L/\rho_k})$.

More formally,
for {\em each} outgoing edge $e = (i,j)\in \se_k$ of capacity $z_e$, a 
$\rho_k \times z_e  $ matrix $\bf C_e$ is specified as a part of the
algorithm. Entries in $\bf C_e$ are chosen from $GF(2^{L/\rho_k})$.
Node $i$ sends to node $j$ a vector $\bf Y_e$ of $z_e$ symbols obtained
as the matrix product ${\bf Y_e} = {\bf X_i C_e}$.
Each element of $\bf Y_e$ is said to be a ``coded symbol''.
The choice of the matrix $\bf C_e$ affects the correctness of the algorithm,
as elaborated later.

\item On each incoming edge $d=(j,i)\in \se_k$, node $i$ receives a vector $\bf Y_d$ containing $z_d$ symbols from $GF(2^{L/\rho_k})$. Node $i$ then checks, for each incoming edge $d$,
whether ${\bf Y_d}={\bf X_i C_d }$. The check is said to fail iff ${\bf Y_d} \neq {\bf X_i C_d }$.

\item If checks of symbols received on any incoming edge fail in the previous step,
then node $i$ sets a 1-bit {\em flag} equal to MISMATCH; else the {\em flag} is set to NULL.
This flag is broadcast in {\em Step 2.2} above.
\end{enumerate}
\end{algorithm}

In the {\em Equality Check} algorithm,
$z_e$ symbols of size $L/\rho_k$ bits are transmitted on
each link $e$ of capacity $z_e$. Therefore, the {\em Equality Check} algorithm
requires time duration 
\begin{eqnarray}
&L~/~\rho_k& \label{e_phase3_time}
\end{eqnarray}

\subsubsection*{Salient Feature of Equality Check Algorithm}

In the {\em Equality Check} algorithm, a single round of communication occurs between
adjacent nodes. No node is required to forward packets received from other
nodes during the {\em Equality Check} algorithm.
This implies that, while a faulty node may send incorrect packets to its neighbors,
it cannot tamper information sent between fault-free nodes.
This feature of {\em Equality Check} is important in being able to prove
its correctness despite the presence of faulty nodes in $\sg_k$.

\subsubsection*{Choice of Parameter $\rho_k$}


We define a set $\Omega_k$ as follows using the disputes identified through
the first $(k-1)$ instances of NAB.
\begin{eqnarray*}
\Omega_k & = & \{ ~ H~~|~~\mbox{$H$ is a subgraph of $\sg_k$ containing $(n-f)$ nodes
such that no two nodes in $H$} \\
&&\mbox{~~~~~~~~~~~ have been found {\em in dispute} through the first
	$(k-1)$ instances of NAB} ~\}
\end{eqnarray*}
As noted in the discussion of Phase 3 (Dispute Control),
fault-free nodes are never found in dispute {\em with
each other} (fault-free nodes may be found in dispute with faulty
nodes, however). This implies that $\sg_k$ includes
all the fault-free nodes, since a fault-free node will
never be found in dispute with $f+1$ other nodes.
There are at least $n-f$ fault-free
nodes in the network. This implies that set $\Omega_k$
is non-empty.

Corresponding to a directed graph $H(V,E)$, let 
us define an undirected graph $\overline H(V,\overline E)$
as follows: (i) both $H$ and $\overline H$ contain the same set
of vertices, (ii) undirected edge $(i,j) \in \overline E$ if
either $(i,j)\in E$ or $(j,i)\in E$, and (iii) capacity of undirected
edge $(i,j)\in\overline E$ is defined to be equal to the sum
of the capacities of directed links $(i,j)$ and $(j,i)$ in $E$ (if
a directed link does not exist in $E$, here we treat its capacity as 0).
For example, Figure \ref{fig:undirected} shows the undirected graph
corresponding to the directed graph in Figure \ref{fig:directed}.

Define a set of {\em undirected} graphs $\overline\Omega_k$ as follows. $\overline\Omega_k$ contains undirected version of each directed graph in $\Omega$.
\[
\overline\Omega_k ~ = ~ \{~~ \overline H ~~ | ~~ H\in \Omega_k ~\}
\]
Define $U_k= \min_{\overline{H}\in \overline{\Omega}_k}\min_{i,j \in\overline{H}}MINCUT(\overline{H},i,j)$ as the minimum value of the MINCUTs between all
 pairs of nodes in
all the undirected graphs in the set $\overline\Omega_k$.
For instance, suppose that $n=4$, $f=1$ and
the graph shown in Figure \ref{fig:original} is $\sg$,
whereas $\sg_k$ is the graph shown in 
Figure \ref{fig:new}. Thus, nodes 2 and 3 have been found in dispute
previously.
Then, $\Omega_k$ and $\overline{\Omega}_k$ each contain two subgraphs,
one subgraph corresponding to the node set $\{1,2,4\}$, 
and the other subgraph corresponding to the node set
$\{1,3,4\}$. In this example, $U_k = 2$. Also notice that in this example, there is no edge between nodes 2 and 4 in $\sg$ to begin with -- so these two nodes will never be found in dispute.

Parameter $\rho_k$ is chosen such that
\[
\rho_k~\leq~ \frac{ U_k}{2}
\]
Under the above constraint on $\rho_k$,
as per (\ref{e_phase3_time}), execution time of {\em Equality Check}
is minimized when $\rho_k=\frac{U_k}{2}$.
Under the above constraint on $\rho_k$, we will prove the correctness of 
the {\em Equality Check} algorithm, with its execution time being $L/\rho_k$.

\subsection{Correctness of the Equality Check Algorithm}
\label{subsec:correctness}
The correctness of Algorithm \ref{alg:MEQ:new} depends on the choices of the parameter $\rho_k$ and the set of coding matrices $\{{\bf C_e}|e\in \se_k\}$. Let us say that a set of coding matrices is {\em correct} if the resulting {\em Equality Check} Algorithm \ref{alg:MEQ:new} satisfies the following requirement:
\begin{itemize}
\item (EC)~~
{\bf if}~~
there exists a pair of fault-free nodes $i,j\in \sg_k$ such
that $\bf X_i\neq X_j$ (i.e., $x_i\neq x_j$), \\
{\bf then}~~
 the 1-bit flag at {\bf at least one} fault-free node is set to MISMATCH.
\end{itemize}
Recall that $\bf X_i$ is a vector representation of the $L$-bit value $x_i$ received by node $i$ in Phase 1 of NAB. Two consequences of the above correctness condition are:
\begin{itemize}
\item If some node (possibly the source node) misbehaves during Phase 1 leading
to outcomes (ii) or (iv) for Phase 1, then at least one fault-free node will
set its flag to MISMATCH. In this case, the fault-free nodes (possibly
including the sender) do not share identical $L$-bit values $\bf X_i$'s as the outcome of Phase 1. 

\item If no misbehavior occurs in Phase 1 (thus the
values received by fault-free nodes in Phase 1 are correct),
but MISMATCH flag at some fault-free node is set in {\em Equality Check},
then misbehavior must have occurred in Phase 2.

\end{itemize}

The following theorem shows that when $\rho_k\le U_k/2$, and when $L$ is sufficiently large, there exists a set coding matrices $\{{\bf C_e}|e\in \se_k\}$ that are correct.

\begin{theorem}
\label{thm:fault-detection}
For $\rho_k \le U_k/2$, when the entries of the coding matrices
 $\{{\bf C_e}|e\in \se_k\}$ in step 1 of Algorithm \ref{alg:MEQ:new} are
 chosen independently and uniformly at random from $GF(2^{L/\rho_k})$,
then $\{{\bf C_e}|e\in \se_k\}$ is correct
 with probability 
$
\geq  ~
1 - 2^{-L/\rho_k} \left[
{n\choose n-f}(n-f-1)\rho_k
\right]. 
$
Note that when $L$ is large enough,  $1 - 2^{-L/\rho_k} \left[
{n\choose n-f}(n-f-1)\rho_k
\right]>0$.
\end{theorem}

\begin{proofSketch}
The detailed proof is presented in Appendix \ref{app:fault-detection}.
Here we provide a sketch of the proof.
%
%
The goal is to prove that property (EC) above holds
with a non-zero probability. That is,
regardless of which (up to $f$) nodes in $\sg$ are faulty,
 when $\bf X_i\neq X_j$ for some pair of fault-free nodes $i$ and $j$
in $\sg_k$ during the $k$-th instance,
at least one fault-free node (which may be different from nodes $i$ and $j$)
will set its 1-bit flag to MISMATCH.
To prove
this, we consider every subgraph of $H\in \Omega_k$ (see definition of $\Omega_k$ above). By definition of $\Omega_k$, no two nodes in $H$ have been found in dispute through the first $(k-1)$ instances of NAB. Therefore, $H$ represents one {\em potential} set of $n-f$ fault-free nodes in $\sg_k$. For each edge $e=(i,j)$ in $H$, steps 1-2 of Algorithm \ref{alg:MEQ:new} together have the effect of  checking whether or not  $\bf (X_i -X_j)  C_e = 0$. Without loss of generality, for the purpose of this proof, {\em rename} the nodes in $H$ as $1,\cdots, n-f$. Denote $\bf D_i = X_i - X_{n-f}$ for $i=1,\cdots, (n-f-1)$, then 
\begin{equation}
\bf (X_i -X_j)  C_e = 0 \Leftrightarrow 
\left\{
\begin{matrix}
\bf (D_i - D_j) C_e = 0 &, & \mbox{~if~}~i,j< n-f;\\
\bf D_i C_e = 0 &, & \mbox{~if~}~j=n-f;\\
\bf -D_j C_e = 0 &, & \mbox{~if~}~i = n-f.
\end{matrix}
\right.
\end{equation}

Define $\bf D_H = [D_1, D_2,\cdots,D_{n-f-1}]$.
Let $m$ be the sum of the capacities of all the directed edges in $H$.
As elaborated in Appendix \ref{app:fault-detection},  we define $\bf C_H$ to be 
a $ (n-f-1)\rho_k \times m$ matrix whose entries are obtained using the 
elements of $\bf C_e$ for each edge $e$ in $H$ in an appropriate manner.
For the suitably defined $\bf C_H$ matrix, we can show that 
the comparisons in steps 1-2 of Algorithm \ref{alg:MEQ:new} at all the
nodes in $H\in \Omega_k$ are equivalent to checking whether or not 
\begin{equation}
\bf D_H \, C_H ~ = ~ 0.
\label{eq:compare-matrix}
\end{equation}
We can show that for a particular subgraph $H\in\Omega_k$, when $\rho_k\le U_k/2$, $m\ge (n-f-1) \rho_k$; and when the set of coding matrices $\{{\bf C_e}|e\in \se_k\}$ are generated as described in Theorem \ref{thm:fault-detection},  for large enough $L$, with non-zero probability $\bf C_H$ contains a $(n-f-1)\rho_k \times (n-f-1)\rho_k$ square submatrix that is invertible. In this case $\bf D_H C_H=0$ if and only if $\bf D_H=0$, i.e., $\bf X_1=X_2=\cdots=X_{n-f}$. In other words, if all nodes in subgraph $H$ are fault-free, and $\bf X_i\neq X_j$ for two fault-free nodes $i,j$, then $\bf D_H C_H\neq 0$ and hence the check in step 2 of Algorithm \ref{alg:MEQ:new} fails at some fault-free node in $H$. 

We can then show that, for large enough $L$,
with a non-zero probability,  this is also
 {\bf simultaneously} true for all subgraphs $H\in\Omega_k$.
This implies that, for large enough $L$,
correct coding matrices ($\bf C_e$ for each $e\in\se_k$) can be found.
These coding matrices are specified as a part of the algorithm
specification. Further details of the proof are in Appendix \ref{app:fault-detection}.
\end{proofSketch}

\section{Correctness of NAB}

For Phase 1 (Unreliable Broadcast) and Phase 3 (Dispute Control),
the proof that the outcomes claimed in Section \ref{sec:overview}
indeed occur follows directly from the prior literature cited in
Section \ref{sec:overview} (and elaborated in Appendices \ref{app:unreliable_broadcast} and \ref{app:diagnosis}). Now consider two cases:
\begin{itemize}
\item
The values received by the
fault-free nodes in Phase 1 are {\em not identical}:  Then the
correctness of {\em Equality Check} ensures that a fault-free
node will detect the mismatch, and consequently Phase 3 will
be performed. As a byproduct of {\em Dispute Control} in Phase 3,
the fault-free nodes will correctly agree on a value that satisfies
the {\em validity} and {\em agreement} conditions.

\item The values received by the fault-free nodes in Phase 1 are identical:
If no node announces a mismatch in step 2.2, then
the fault-free nodes will agree on the value received in Phase 1. It is
easy to see that this is a correct outcome. On the other hand, if
some (faulty) node announces a mismatch in step 2.2, then {\em Dispute
Control} will be performed, which will result in correct outcome
for the broadcast of the $k$-th instance.
\end{itemize}
Thus, in all cases, NAB will lead to correct outcome in each instance.

\section{Throughput of NAB and Capacity of BB}
\label{sec:throughput}

\subsection{A Lower Bound on Throughput of NAB for Large $Q$ and $L$}
\label{subsec:throughput_NAB}

In this section, we provide the intuition behind the derivation of the
lower bound. More detail is presented in
Appendix \ref{app:BB:throughput}.
We prove the lower bound when the number of instances $Q$
and input size $L$ for each instance are both
``large'' (in an order sense) compared to $n$.
Two consequences of $L$ and $Q$ being large:
\begin{itemize}
\item As a consequence of $Q$ being large,
the average overhead of {\em Dispute control} per instance of NAB becomes
negligible. Recall that
{\em Dispute Control} needs to be performed at most $f(f+1)$ times over
$Q$ executions of NAB.

\item As a consequence of $L$ being large, the
overhead of 1-bit broadcasts performed in step 2.2 of Phase 2
becomes negligible when amortized over the $L$ bits being broadcast
by the source in each instance of NAB.
\end{itemize}
It then suffices to consider only the time it takes to complete the
{\em Unreliable broadcast} in Phase 1 and {\em Equality Check} in Phase 2. 
For the $k$-th instance of NAB, as discussed previously, the unreliable broadcast in Phase 1 can be done in $L/\gamma_k$ time units (see definition of $\gamma_k$ in section \ref{sec:overview}.). We now define 
\begin{eqnarray*}
\Gamma & = & \{ ~ H~~|~~\mbox{$H$ is a subgraph of $\sg$ containing source node 1,} \\
&&\mbox{~~~~~~~~~~~ and $\sg_k$ may equal $J$ in some execution of NAB for some $k$} ~\}\end{eqnarray*} 
Appendix \ref{app:subgraph} provides a systematic construction of the set $\Gamma$. Define the minimum value of all possible $\gamma_k$:
$$
\gamma^* = \min_{\sg_k\in \Gamma}\gamma_k=\min_{\sg_k\in\Gamma}~\min_{j\in \sv_k}~MINCUT(\sg_k,1,j).
$$ 
Then an upper bound of the execution time of Phase 1 in all instances of NAB is $L/\gamma^*$.

With parameter $\rho_k= U_k/2$, the execution time of the Equality Check in Phase 2 is $L/\rho_k$. Recall that $U_k$ is defined as the minimum value of the MINCUTs between all pairs of nodes in all undirected graphs in the set $\overline\Omega_k$. As discussed in Appendix \ref{app:BB:all_invertible}, $\overline{\Omega}_k\subseteq \overline{\Omega}_1$, where $\sg_1 = \sg$. Hence $U_k\ge U_1$ in all possible $\sg_k$. Define 
$$
\rho^* = \frac{U_1}{2} = \min_{\overline{H}\in \Omega_1} \min_{nodes~i,j ~in~\overline{H}}MINCUT(\overline{H},i,j).
$$
Then $\rho_k\ge \rho^*$ for all possible $\sg_k$ and the execution time of the Equality Check is upper-bounded by $L/\rho^*$.
So the throughput of NAB for large $Q$ and $L$ can be lower bounded
by\footnote{\normalsize To simplify the analysis above, we ignored propagation delays.
Appendix \ref{app:BB:throughput} describes how to achieve this bound
even when propagation delays are considered.}
\begin{equation}
\lim_{L\rightarrow \infty} T(\sg,L,NAB)\ge \frac{L}{L/\gamma^*+L/\rho^*} = \frac{\gamma^*\rho^*}{\gamma^*+\rho^*}.
\label{eq:throughput}
\end{equation}


\subsection{An Upper Bound on Capacity of BB}
\label{sec:upper}

\begin{theorem}
\label{thm:BB:bound}
In any point-to-point network $\sg(\sv,\se)$, the capacity of Byzantine broadcast ($C_{BB}$) with node 1 as the source
satisfies the following upper bound
$$C_{BB}(\sg) \le \min(\gamma^*,2\rho^*).$$
\end{theorem}

Appendix \ref{app:upperbound} presents a proof of this upper bound.
%
%
%
Given the throughput lower bound $T_{NAB}(\sg)$ in (\ref{eq:throughput}) and the
upper bound on $C_{BB}(\sg)$ from Theorem \ref{thm:BB:bound}, as
shown in Appendix \ref{app:fraction}, the result below can be
obtained.

\begin{theorem}
\label{thm:fraction}
For graph $\sg(\sv,\se)$:
$$
\lim_{L\rightarrow \infty} T(\sg,L,NAB)\ge \min(\gamma^*,2\rho^*)/3 \ge C_{BB}(\sg)/3.
$$
Moreover, when $\gamma^*\le \rho^*$:
$$
\lim_{L\rightarrow \infty} T(\sg,L,NAB)\ge \min(\gamma^*,2\rho^*)/2 \ge C_{BB}(\sg)/2.
$$
\end{theorem}

\section{Conclusion}
This paper presents NAB, a network-aware Byzantine broadcast algorithm for point-to-point networks. We derive an upper bound on the capacity of Byzantine broadcast, and show that NAB can achieve throughput at least 1/3 fraction of the capacity over a large number of execution instances, when $L$ is large. The fraction can be improved to at least 1/2 when the network satisfies an additional condition.

\bibliography{PaperList}

\appendix
\appendixpage

\section{Unreliable Broadcast in Phase 1}
\label{app:unreliable_broadcast}

According to \cite{spanningTreePacking}, in a given graph $\sg_k$ with $\gamma_k =\min_{j\in \sv_k}~MINCUT(\sg_k,1,j)$, there always exist a set of $\gamma_k$ unit-capacity spanning trees of $\sg_k$ such that the total usage on each edge  $e\in \se_k$ by all the $\gamma_k$ spanning trees combined is no more than its link capacity $z_e$.
Each spanning tree is ``unit-capacity'' in the sense that 1 unit capacity
of each link on that tree is allocated for transmissions on that tree.
 For example, Figure \ref{fig:directed-trees} shows 2 unit-capacity spanning trees that can be embedded in the directed graph in Figure \ref{fig:directed}: one spanning tree is shown with solid edges and the other spanning tree is shown in dotted edges.
Observe that  link (1,2) is used by both spanning trees, each tree using a unit capacity
on link (1,2), for a total usage of 2 units, which is the capacity of link (1,2).

To broadcast an $L$-bit value from source node 1,
we represent the $L$-bit value as $\gamma_k$ symbols, each symbol being
represented using $L/\gamma_k$ bits. One symbol ($L/\gamma_k$ bits) is then
transmitted along each of the $\gamma_k$ unit-capacity spanning trees.

\section{Dispute Control}
\label{app:diagnosis}

The dispute control algorithm is performed in the $k$-th instance of NAB
only if at least one node misbehaves during Phases 1 or 2. The goal of dispute
control is to learn some information about the identity of at least one faulty
node. In particular, the dispute control algorithm will identify a new node
as being faulty, or/and identify a new node pair in dispute (at least one of the nodes
in the pair is guaranteed to be faulty). The steps in dispute control
in the $k$-th instance of NAB are as follows:
\begin{itemize}
\item (DC1) Each node $i$ in $\sv_k$ uses a previously proposed 
Byzantine broadcast algorithm, such as \cite{opt_bit_Welch92},
to broadcast to all other nodes in $\sv_k$ all the messages
that this node $i$ claims to have received from
other nodes, and sent to the other nodes, during Phases 1 and 2 of the $k$-th instance. 
Source node 1 also uses an existing Byzantine broadcast algorithm \cite{opt_bit_Welch92}
to broadcast its $L$-bit input for the $k$-th instance to all the other nodes.
Thus, at the end of this step, all the fault-free nodes will reach correct
agreement for the output for the $k$-th instance.

\item (DC2) If for some node pair $a,b\in\sv_k$, a message that node $a$
claims above to have sent to node $b$ mismatches with the claim of received
messages made by node $b$, then node pair $a,b$ is found in dispute.
In step DC1, since a Byzantine broadcast algorithm is used to disseminate the
claims, all the fault-free nodes will identify identical node pairs in dispute.

It should be clear that a pair of fault-free nodes will never be found in
dispute with each other in this step.

\item (DC3) The NAB algorithm is deterministic in nature. Therefore, the messages
that should be sent by each node in Phases 1 and 2 can be completely determined by
the messages that the node receives, and, in case of node 1, its initial input.
Thus, if the claims of the messages sent by some node $i$ are inconsistent
with the message it claims to have received, and its initial input (in case of node 1),
then that node $i$ must be faulty. Again, all fault-free nodes identify these
faulty nodes identically. Any nodes thus identified as faulty until now (including
all previous instances of NAB) are deemed to be ``in dispute'' with all
their neighbors (to whom the faulty nodes have incoming or outgoing links).

It should be clear that a fault-free node will never be found to be
faulty in this step.

\item (DC4) Consider the node pairs that have been identified
as being in dispute in DC2 and DC3 of at least one instances of NAB so far.

We will say that a set of nodes $F_i$, where $|F_i|\leq f$, ``explains'' all the
disputes so far, if for each pair $a,b$ found in dispute so far, at least one of
$a$ and $b$ is in $F_i$.
It should be easy to see that for any set of disputes that may be observed, there
must be at least one such set that {\em explains} the disputes.
It is easy to argue that the nodes in the set below must be necessarily faulty
(in fact, the nodes in the set intersection below are also guaranteed to include
nodes identified as faulty in step DC3).
$$\bigcap_{\delta=1}^{\Delta} F_\delta$$
Then, $\sv_{k+1}$ is obtained as $\sv_k-\bigcap_{\delta=1}^{\Delta} F_\delta$.
$\se_{k+1}$ is obtained by removing from $\se_k$ edges incident on nodes 
in $\bigcap_{\delta=1}^{\Delta} F_\delta$, and also excluding edges between node
pairs that have been found in dispute so far.
\end{itemize}

As noted earlier, the above dispute control phase may be executed in at most $f(f+1)$
instances of NAB.

\section{Proof of Theorem \ref{thm:fault-detection}}
\label{app:fault-detection}

To prove Theorem \ref{thm:fault-detection}, we first prove that when the coding matrices are generated at random as described, for a particular subgraph $H\in \Omega_k$, with non-zero probability, the coding matrices $\{{\bf C_e}|e\in \sg_k\}$ defines a matrix $\bf C_H$ (as defined later) such that $\bf D_H C_H=0$ if and only if $\bf D_H = 0$. Then we prove that this is also {\bf simultaneously} true for all subgraphs $H\in \Omega_k$.

\subsection{For a given subgraph $H\in\Omega_k$}
\label{subsec:MEQ:correctness}

Consider any subgraph $H\in\Omega_k$. For each edge $e=(i,j)$ in $H$, we ``expand'' the corresponding coding matrix $\bf C_e$ (of size  $\rho_k \times z_e$) to a $(n-f-1)\rho_k \times z_e$ matrix $\bf B_e$ as follows: $\bf B_e$ consists $n-f-1$ blocks, each block is a $\rho_k \times z_e$ matrix:
\begin{itemize}
\item If $i\neq n-f$ and $j\neq n-f$, then the $i$-th and $j$-th block equal to $\bf C_e$ and $-{\bf C_e}$, respectively. The other blocks are all set to $\bf 0$. 
$$
{\bf B_e}^T = \bordermatrix{~ & &  i & & j &\cr
&   \bf  0 \cdots 0& {\bf C_e}^T & \bf 0 \cdots 0& -{\bf C_e}^T &\bf 0 \cdots  0\cr
}
$$
Here $()^T$ denotes the transpose of a matrix or vector.
\item If $i= n-f$, then the $j$-th block equals to $-{\bf C_e}$, and the other blocks are all set to 0 matrix.
$$
{\bf B_e}^T = \bordermatrix{~ & &  j & \cr
&    \bf 0 \cdots 0& -{\bf C_e}^T & \bf 0 \cdots 0\cr
}
$$
\item If $j= n-f$, then the $i$-th block equals to $\bf C_e$, and the other blocks are all set to 0 matrix.
$$
{\bf B_e}^T = \bordermatrix{~ & &  i & \cr
& \bf 0 \cdots 0& {\bf C_e}^T & \bf 0 \cdots 0\cr
}
$$

\end{itemize}

Let $D_{i,\beta} = X_i(\beta) - X_{n-f}(\beta)$ for $i< n-f$ as the difference between $\bf X_i$ and $\bf X_{n-f}$ in the $\beta$-th element. Recall that $\bf D_i = X_i - X_{n-f} = \begin{pmatrix}D_{i,1}&\cdots&D_{i,\rho_k}\end{pmatrix}$ and ${\bf D_H} = \begin{pmatrix}\bf D_1 & \cdots &\bf D_{n-f-1} \end{pmatrix}$. So $\bf D_H$ is a row vector of $(n-f-1)\rho_k$ elements from $GF(2^{L/\rho_k})$ that captures the differences between $\bf X_i$ and $\bf X_{n-f}$ for all $i<n-f$. It should be easy to see that 
$$
({\bf X_i} - {\bf X_j}){\bf C_e} =\bf 0 \Leftrightarrow \bf D_H B_e  = 0.
$$
So for edge $e$, steps 1-2 of Algorithm \ref{alg:MEQ:new} have the effect of checking whether or not $\bf D_H B_e = 0$.

If we label the set of edges in $H$ as $e1,e2,\cdots$, and let $m$ be the sum of the capacities of all edges in $H$, then we construct a $(n-f-1)\rho_k \times m$ matrix  $\bf C_H$ by concatenating all expanded coding matrices:
$$
{\bf C_H} = \begin{pmatrix}
{\bf B_{e1}} & 
{\bf B_{e2}} &
\cdots 
\end{pmatrix},
$$
where each column of $\bf C_H$ represents one coded symbol sent in $H$ over the corresponding edge.
Then steps 1-2 of Algorithm \ref{alg:MEQ:new} for all edges in $H$ have the same effect of checking whether or not $\bf D_H C_H = 0$. So to prove Theorem \ref{thm:fault-detection}, we need to show that there exists at least one $\bf C_H$ such that $$\bf D_H C_H = 0 ~\Leftrightarrow~ D_H = 0.$$

It is obvious that if $\bf D_H = 0$, then $\bf D_H C_H = 0$ for any $\bf C_H$. So all left to show is that there exists at least one $\bf C_H$ such that $\bf D_H C_H = 0 \Rightarrow D_H=0$. It is then sufficient to show that $\bf C_H$ contains a $(n-f-1)\rho_k \times (n-f-1)\rho_k$ submatrix $\bf M_H$ that is {\bf invertible}, because when such an invertible submatrix exist,
$$\bf D_H C_H = 0 ~\Rightarrow ~ D_H M_H = 0 ~\Rightarrow~D_H = 0.$$

Now we describe how one such submatrix $\bf M_H$ can be obtains. 
Notice that each column of $\bf C_H$ represents one coded symbol sent on the corresponding edge. A $(n-f-1)\rho_k \times (n-f-1)$ submatrix $\bf S$ of $\bf C_H$ is said to be a ``spanning matrix'' of $H$ if the edges corresponding to the columns of $\bf S$ form a undirected spanning tree of $\overline{H}$ -- the {\em undirected} representation of $H$. In Figure \ref{fig:undirected-tree}, an undirected spanning tree of the undirected graph in Figure \ref{fig:undirected} is shown in dotted edges. It is worth pointing out that an undirected spanning tree in an undirected graph $\overline{H}$ does not necessarily correspond to a directed spanning tree in the corresponding directed graph $H$. For example, the directed edges in Figure \ref{fig:directed} corresponding to the dotted undirected edges in Figure \ref{fig:undirected-tree} do not consist a spanning tree in the directed graph in Figure \ref{fig:directed}.

It is known that in an undirected graph whose MINCUT equals to $U$, at least $U/2$ undirected unit-capacity spanning trees can be embedded \cite{spanningTreePacking}\footnote{\normalsize The definition of embedding undirected unit-capacity spanning trees in undirected graphs is similar to embedding directed unit-capacity spanning trees in directed graphs (by dropping the direction of edges).} This implies that $\bf C_H$ contains a set of $U_k/2$ spanning matrices such that no two spanning matrices in the set covers the same column in  $\bf C_H$. Let $\{\bf S_1,\cdots,S_{\rho_k}\}$ be one set of $\rho_k\le U_k/2$ such spanning matrices of $H$. Then union of these spanning matrices forms an $(n-f-1)\rho_k \times (n-f-1)\rho_k$ submatrix of $\bf C_H$:
$$
\bf M_H = \begin{pmatrix}
\bf S_1 &
\cdots &
\bf S_{\rho_k}
\end{pmatrix}.
$$

Next, we will show that when the set of coding matrices are generated as described in Theorem \ref{thm:fault-detection}, with non-zero probability we obtain an invertible square matrix $\bf M_H$. When $\bf M_H$ is invertible, 
$$\bf D_H M_H =0 ~\Leftrightarrow~ D_H = 0 ~\Leftrightarrow~ X_1=\cdots = X_{n-f}.$$

For the following discussion, it is convenient to reorder the elements of $\bf D_H$ into 
$${\bf \tilde{D}_H} = \begin{pmatrix}D_{1,1} & \cdots& D_{n-f-1,1}&D_{1,2}& \cdots& D_{n-f-1,2}&\cdots& D_{1,\rho_k}& \cdots& D_{n-f-1,\rho_k}\end{pmatrix},$$
so that the $(\beta-1)(n-f-1)+1$-th through the $\beta(n-f-1)$ elements of $\bf \tilde{D}_H$ represent the difference between $\bf X_i$ ($i=1,\cdots,n-f-1$) and $\bf X_{n-f}$ in the $\beta$-th element.

We also reorder the rows of each spanning matrix $\bf S_q$ ($q=1,\cdots,\rho_k$) accordingly. It can be showed that after reordering, $\bf S_q$ becomes $\bf \tilde{S}_q$ and has the following structure:
\begin{equation}
\bf \tilde{S}_q = 
\begin{pmatrix}
\bf A_q S_{q,1}  \\
\bf A_q S_{q,2} \\
 \vdots \\
\bf A_q S_{q,\rho_k}
\end{pmatrix}.
\end{equation}

Here $\bf A_q$ is a $(n-f-1)\times (n-f-1)$ square matrix, and it is called the {\em adjacency} matrix of the spanning tree corresponding to $\bf S_q$. $\bf A_q$ is formed as follows. Suppose that the $r$-th column of $\bf S_q$ corresponds to a coded symbol sent over a directed edge $(i,j)$ in $H$, then 
\begin{enumerate}
\item If $i\neq n-f$ and $j\neq n-f$, then the $r$-th column of $\bf A_q$ has the $i$-th element as 1 and the $j$-th element as -1, the remaining entries in that column are all 0;

\item If $i=n-f$, then the $j$-th element of the $r$-th column of $\bf A_q$ is set to -1, the remaining elements of that column are all 0;

\item If $j=n-f$, then the $i$-th element of the $r$-th column of $\bf A_q$ is set to 1, the remaining elements of that column are all 0.
\end{enumerate}
For example, suppose $\overline{H}$ is the graph shown in Figure \ref{fig:undirected}, and $\bf S_q$ corresponds to a spanning tree of $H$ consisting of the dotted edges in Figure \ref{fig:undirected-tree}. Suppose that we index the corresponding directed edges in the graph shown in Figure \ref{fig:directed} in the following order: (2,3), (1,4), (4,3). The resulting adjacency matrix $\bf A_q =  
\begin{pmatrix}
0 & 1 & 0 \\
1 & 0 & 0 \\
-1 & 0 & -1 
\end{pmatrix}$.

On the other hand, each  $(n-f-1)\times (n-f-1)$ square matrix $\bf S_{q,p}$ is a diagonal matrix. The $r$-th diagonal element of $\bf S_{q,p}$ equals to the $p$-th coefficient used to compute the coded symbol corresponding to the $r$-th column of $\bf S_q$. For example, suppose the first column of $\bf S_q$ corresponds to a coded packet $ X_1(1) + 2X_1(2)$ being sent on link $(1,2)$. Then the first diagonal elements of $\bf S_{q,1}$ and $\bf S_{q,2}$ are 1 and 2, respectively. 

So after reordering, $\bf M_H$ can be written as $\bf \tilde{M}_H$ that has the following structure:
\begin{equation}
\bf \tilde{M}_H = \begin{pmatrix}
\bf A_1 S_{1,1}  & \bf A_2 S_{2,1} & \cdots & \bf A_{\rho_k}C_{\rho_k,\rho_k} \\
\bf A_1 S_{1,2} & \bf A_2 S_{2,2}  & \cdots & \bf A_{\rho_k}C_{\rho_k, \rho_k} \\
\vdots&& \ddots&\vdots\\
\bf A_1 S_{1,\rho_k} & \bf A_2 S_{2,\rho_k} & \cdots & \bf A_{\rho_k} S_{\rho_k,\rho_k} 
\end{pmatrix}
\end{equation}
Notice that $\bf \tilde{M}_H$ is obtained by permuting the rows of $\bf M_H$. So to show that $\bf M_H$ being invertible is equivalent to $\bf \tilde{M}_H$ being invertible.

Define 
$
\bf M_q = 
\begin{pmatrix}
\bf A_1 S_{1,1} &  \cdots & \bf A_q S_{q,1}\\
\vdots  & \ddots & \vdots\\
\bf A_1 S_{1,q} &  \cdots & \bf A_q S_{q,q}
\end{pmatrix}
$ for $1\le q\le \rho_k$. Note that $\bf M_{q1}$ is a sub-matrix of $\bf M_{q2}$ when $q1< q2$, and $\bf M_{\rho_k} = \tilde{M}_H$. We prove the following lemma:

\begin{lemma}
\label{lm:single_H}
For any $\rho_k\le U_k/2$, with probability at least $\left(1-\frac{n-f-1}{2^{L/\rho_k}}\right)^{\rho_k}$, matrix $\bf \tilde{M}_H$ is invertible. Hence $\bf M_H$ is also invertible.
\end{lemma}
\begin{proof}
We now show that each $\bf M_q$ is invertible with probability at least $\left(1-\frac{n-f-1}{2^{L/\rho_k}}\right)^q$ for all $q\le \rho_k$. The proof is done by induction, with $q=1$ being the base case.

\paragraph{Base Case -- $q=1$:}
\begin{equation}
\bf M_1 =  A_1 S_{1,1}.
\end{equation}
As showed later in Appendix \ref{app:A:invertible}, $\bf A_q$ is always invertible and $\det(\bf A_q)=\pm 1$. 
Since $\bf S_{1,1}$ is a $(n-f-1)$-by-$(n-f-1)$ diagonal matrix, it is
invertible provided that all its $(n-f-1)$ diagonal elements are non-zero. Remember that the diagonal elements of $\bf S_{1,1}$ are chosen uniformly and independently  from
$GF(2^{L/\rho_k})$. The probability that they are all non-zero is $\left(1-\frac{1}{2^{L/\rho_k}}\right)^{n-f-1}\ge 1- \frac{n-f-1}{2^{L/\rho_k}}$.

\paragraph{Induction Step -- $q$ to $q+1\le \rho_k$:}
%
%
The square matrix $\bf M_{q+1}$ can be written as
\begin{equation}
\bf M_{q+1} = \begin{pmatrix}
\bf M_q & \bf P_q\\
\bf F_q & \bf A_{q+1}S_{q+1,q+1}
\end{pmatrix},
\end{equation}
where
\begin{equation}
\bf P_q = 
\begin{pmatrix}
\bf A_{q+1} S_{q+1,1}\\
\bf A_{q+1} S_{q+1,1}\\
\vdots\\
\bf A_{q+1} S_{q+1,q}
\end{pmatrix}
\end{equation}
is an $(n-f-1)q$-by-$(n-f-1)$ matrix, and 
\begin{equation}
\bf F_q = 
\begin{pmatrix}
\bf A_1 S_{1,k+1} & \cdots & \bf A_q S_{q,q+1}
\end{pmatrix}
\end{equation}
is an $(n-f-1)$-by-$(n-f-1)q$ matrix.

Assuming that $\bf M_q$ is invertible, we transform $\bf M_{q+1}$ as follows:
\begin{eqnarray}
\bf M_{q+1}' &=& 
\begin{pmatrix}
\bf I_{(n-f-1)q} & \bf 0\\
\bf 0 & \bf A_{q+1}^{-1}
\end{pmatrix}
\bf M_{k+1}
\begin{pmatrix}
\bf I_{(n-f-1)q} & \bf - M_q^{-1} P_q\\
\bf 0 & \bf I_{(n-f-1)}
\end{pmatrix}
\\
& =&
\begin{pmatrix}
\bf I_{(n-f-1)q} & \bf 0\\
\bf 0 & \bf A_{q+1}^{-1}
\end{pmatrix}
\begin{pmatrix}
\bf M_q & \bf P_q\\
\bf F_q & \bf A_{q+1}S_{q+1,q+1}
\end{pmatrix}
\begin{pmatrix}
\bf I_{(n-f-1)q} & \bf - M_q^{-1} P_q\\
\bf  0 & \bf I_{(n-f-1)}
\end{pmatrix}
\\
&=&
\begin{pmatrix}
\bf M_q & \bf 0\\
\bf A_{q+1}^{-1}F_q  & \bf S_{q+1,q+1} - A_{q+1}^{-1}F_q M_q^{-1} P_q 
\end{pmatrix}.
\end{eqnarray}
Here $\bf I_{(n-f-1)q}$ and $\bf I_{(n-f-1)}$ each denote a $(n-f-1)q\times (n-f-1)q$ and a $(n-f-1)\times (n-f-1)$ identity matrices. Note that 
$
|\det(\bf M'_{k+1})|=|\det(\bf M_{k+1})|,
$
since the matrix multiplied at the left has determinant $\pm 1$, and the matrix multiplied at the right has determinant 1.

Observe that the diagonal elements of the $(n-f-1)\times (n-f-1)$ diagonal matrix $\bf S_{q+1,q+1}$ are chosen independently from $\bf A_{q+1}^{-1}F_q M_q^{-1} P_q$. Then it can be proved that $\bf S_{q+1,q+1} - A_{q+1}^{-1}F_q M_q^{-1} P_q$ is invertible with probability at least $1-\frac{n-f-1}{2^{L/\rho_k}}$ (See Appendix \ref{app:M:invertible}.) given that $\bf M_q$ is invertible, which happens with probability at least $\left(1-\frac{n-f-1}{2^{L/\rho_k}}\right)^q$ according to the induction assumption. So we have
\begin{eqnarray}
\Pr\{{\bf M_{q+1}}~\mathrm{is~invertible}\} 
\ge \left(1-\frac{n-f-1}{2^{L/\rho_k}}\right)^q \left(1-\frac{n-f-1}{2^{L/\rho_k}}\right)
=\left(1-\frac{n-f-1}{2^{L/\rho_k}}\right)^{q+1}.
\end{eqnarray}
This completes the induction.
Now we can see that $\bf M_{\rho_k} = \tilde{M}_H$ is invertible with probability
\begin{equation}
\ge~~ \left(1- \frac{n-f-1}{2^{L/\rho_k}}\right)^{\rho_k} 
~\ge~ 1- \frac{(n-f-1)\rho_k}{2^{L/\rho_k}} 
~\rightarrow~ 1,~\mathrm{as}~L\rightarrow \infty.
\end{equation}
\end{proof}

Now we have proved that there exists a set of coding matrices $\{{\bf C_e}|e\in \se_k\}$ such that the resulting  $\bf C_H$ satisfies the condition that $\bf D_H C_H = 0$ if and only if $\bf D_H = 0$.

\subsection{For all subgraphs in $\Omega_k$}
\label{app:BB:all_invertible}
\setcounter{theorem}{3}

\comment{
\begin{theorem}
When performing  Algorithm \ref{alg:MEQ:new} with parameter $\rho=R^*$ in $\sg = G'$, all matrices $M^{(k)}$ are invertible simultaneously with probability at least $1-{n\choose n-f}(n-f-1)R^*/2^{B/R^*}$.
\end{theorem}
}

In this section, we are going to show that, for $\sg_k$, if the coding matrices $\{{\bf C_e}|e\in \se_k\}$ are generated as described in Theorem \ref{thm:fault-detection}, then with non-zero probability the set of square matrices $\{{\bf M_H}|H\in \Omega_k\}$ are all invertible {\bf simultaneously}. When this is true, there exists a set of coding matrices that is correct. 

To show that $\bf M_H$'s for all $H\in \Omega_k$ are simultaneously invertible with non-zero probability, we consider the product of all these square matrices:
$$\prod_{H\in\Omega_k} \bf {M_H}.$$
According to Lemma \ref{lm:single_H},  each $\bf M_H$ ($H\in \Omega_k$) is invertible with non-zero probability. It implies that $\det(\bf M_H)$ is a non-identically-zero polynomial of the random coding coefficients of degree at most $(n-f-1)\rho_k$ (Recall that $\bf M_H$ is a square matrix of size $(n-f-1)\rho_k$.). 
So 
$$\det\left(\prod_{H\in\Omega_k} \bf {M_H}\right) 
= \prod_{H\in \Omega_k}\det\left(\bf M_H\right)$$
is a non-identically-zero polynomial of the random coefficients of degree at most $|\Omega_k|(n-f-1)\rho_k$. Notice that each coded symbol is used once in each subgraph $H$. So each random coefficient appears in at most one column in each $\bf M_H$. It follows that the largest exponent of any random coefficient in $\det\left(\prod_{H\in\Omega_k} \bf {M_H}\right) $ is at most $|\Omega_k|$.

According to Lemma 1 of \cite{Tracy:randomCoding}, the probability that $\det\left(\prod_{H\in\Omega_k} \bf {M_H}\right) $ is non-zero is at least
$$
\left(1-2^{-L/\rho_k}|\Omega_k|\right)^{(n-f-1)\rho_k}
\ge 1-2^{-L/\rho_k}\left[|\Omega_k|(n-f-1)\rho_k\right]
.$$
According to the way $\sg_k$ is constructed and the definition of $\Omega_k$, it should not be hard to see that $\sg_k$ is a subgraph of $\sg_1=\sg$, and $\Omega_k\subseteq \Omega_1$. Notice that $|\Omega_1|= {n\choose n-f}$. So $|\Omega_k|\le {n\choose n-f}$ and Theorem \ref{thm:fault-detection} follows.

\subsection{Proof of $\bf A_q$ being Invertible}
\label{app:A:invertible}
Given an adjacency matrix $\bf A_q$, let us call the corresponding spanning tree of $\overline{H}$ as $T_q$. For edges in $T_q$ incident on node $n-f$, the corresponding columns in $\bf A_q$ have exactly
one non-zero entry. Also, the column corresponding to an edge that is incident on
node $i$ has a non-zero entry in row $i$.	
Since there must be at least one edge in $T_q$ that is incident on node $n-f$, there must be at least one column of $\bf A_q$ that has only one non-zero element. Also, since every node is incident on at least one edge in $T_q$, every row of $\bf A_q$ has at least one non-zero element(s). Since there is at most one edge between every pair of nodes in $T_q$, no two columns in $\bf A_q$ are non-zero in identical rows. Therefore, by column manipulation, we can transform matrix $\bf A_q$ into another matrix in which every row and every column has exactly one non-zero element. Hence $\det(\bf A_q)$ equals to either $1$ or $-1$, and $\bf A_q$ is invertible.

\subsection{Proof of $\bf S_{q+1,q+1} - A_{q+1}^{-1}F_q M_q^{-1} P_q$ being Invertible}
\label{app:M:invertible}
Consider $\bf W$ be an arbitrary {\em fixed} $w\times w$ matrix. Consider a random $w\times w$ diagonal matrix $\bf S$
with $w$ diagonal elements $s_1,\cdots,s_w$.
\begin{equation}
\bf S = 
\begin{pmatrix}
s_1 & 0  & \cdots & 0\\
0 & s_2  & \cdots & 0\\
\vdots & &  \ddots  & \vdots\\
0 & \cdots & 0 & s_w
\end{pmatrix}
\end{equation}
The diagonal elements of $\bf S$ are selected independently and uniformly randomly from $GF(2^{\rho_k})$.
Then we have:
\begin{lemma}
\label{thm:invertible}
The probability that the $w\times w$ matrix $\bf S-W$ is invertible is lower bounded by:
\begin{equation}
\Pr\{\textrm{$(\bf S-W)$ is invertible}\} \ge 1-\frac{w}{2^{\rho_k}}.
\end{equation}
\end{lemma}

\begin{proof}
Consider the determinant of matrix $\bf S-W$.
\begin{eqnarray}
\det (\bf S-W) &=& \det
\begin{pmatrix}
(s_1 -  W_{1,1}) & -R_{1,2}  & \cdots & -W_{1,w} \\
-R_{2,1} & (s_2 - W_{2,2})  & \cdots & -R_{2,w}\\
\vdots &  & \ddots & \vdots\\
-R_{w,1} & \cdots & -W_{w,w-1} & (s_r - R_{w,w})
\end{pmatrix}\\
&=& (s_1-W_{1,1})(s_2-W_{2,2})\cdots(s_w - W_{w,w}) ~~ +~~  \mbox{other terms}\\
&=& \Pi_{i=1}^w s_i ~~+~~ W_{w-}
\end{eqnarray}
The first term above, $\Pi_{i=1}^w s_i$, is a degree-$w$ polynomial
of $s_1,\cdots, s_w$. $W_{w-}$ is a polynomial
of degree at most $w-1$ of $s_1,\cdots, s_w$, and it represents
the remaining terms in $\det(\bf S-W)$.
Notice that $\det(\bf S-W)$ cannot be identically zero since it contains
only one degree-$w$ term.
Then by the Schwartz-Zippel Theorem, the probability that $\det(\bf S-R) = 0$ is  $\le w/2^{\rho_k}$. Since $\bf S-W$ is invertible if and only if $\det(\bf S-W) \neq 0$, we conclude that
\begin{equation}
\Pr\{\textrm{$(\bf S-W)$ is invertible}\} \ge 1-\frac{w}{2^{\rho_k}}
\end{equation}

By setting $\bf S=S_{q+1,q+1}$, $\bf W=A_{q+1}^{-1}F_q M_q^{-1} P_q$, and $w=n-f-1$, we prove that $\bf S_{q+1,q+1} - A_{q+1}^{-1}F_q M_q^{-1} P_q$ is invertible with probability at least $1-\frac{n-f-1}{2^{L/\rho_k}}$.
\end{proof}

\section{Throughput of NAB}
\label{app:BB:throughput}

First consider the time cost of each operation in instance $k$ of NAB : 
\begin{itemize}
\item {\bf Phase 1}: It takes $L/\gamma_k \le L/\gamma^*$ time units, since unreliable broadcast from the source node 1 at rate  $\gamma_k$ is achievable and $\gamma_k\ge \gamma^*$, as discussed in Appendix \ref{app:unreliable_broadcast}. 

\item  {\bf Phase 2 -- Equality check}: As discussed previously, it takes $L/\rho_k \le  L/\rho^*$ time units. 


\item  {\bf Phase 2 -- Broadcasting outcomes of equality check}: To reliably broadcast the 1-bit flags from the equality check algorithm, a previously proposed Byzantine broadcast algorithm, such as \cite{opt_bit_Welch92}, is used. The algorithm from \cite{opt_bit_Welch92}, denoted as \BSB hereafter, reliably broadcasts 1 bit by communicating no more than $P(n)$ bits in a {\em complete} graph, where $P(n)$ is a polynomial of $n$. In our setting, $\sg$ might not be complete. However, the connectivity of $\sg$ is at least $2f+1$. It is well-known that, in a graph with connectivity at least $2f+1$ and  at most $f$ faulty nodes, reliable  {\em end-to-end} communication from any node $i$ to any other node $j$ can be achieved by sending the same copy of data along a set of $2f+1$ node-disjoint paths from node $i$ to node $j$ and taking the majority at node $j$. By doing this, we can emulate a complete graph in an incomplete graph $\sg$. Then it can be showed that, by running \BSB on top of the emulated complete graph,  reliably broadcasting the 1-bit flags can be completed in $O(n^\alpha)$ time units, for some constant $\alpha>0$. 

\item  {\bf Phase 3}:  If Phase 3 is performed in instance $k$, every node $i$ in $\sv_k$ uses \BSB to reliably broadcast all the messages that it claims to have received from other nodes, and sent to the other nodes, during Phase 1 and 2 of the $k$-th instance.  Similar to the discussion above about broadcasting the outcomes of equality check, it can be showed that the time it takes to complete Phase 3 is $O(Ln^\beta)$ for some constant $\beta>0$.
\end{itemize}

Now consider a sequence of $Q>0$ instances of NAB. As discussed previously, Phase 3 will be performed at most $f(f+1)$ times throughout the execution of the algorithm. So we have the following upper bound of the execution time of $Q$ instances of NAB:
\begin{equation}
t(\sg,L,Q,NAB) \le Q\left(\frac{L}{\gamma^*} + \frac{L}{\rho^*} + O(n^\alpha)\right) + f(f+1)O(Ln^\beta).
\end{equation}
Then the throughput of NAB can be lower bounded by  
\begin{eqnarray}
T(\sg,L,NAB) &=& \lim_{Q\rightarrow \infty} \frac{LQ}{t(\sg,L,Q,NAB)}\\
&\ge& \lim_{Q\rightarrow \infty} \frac{LQ}{Q\left(\frac{L}{\gamma^*} + \frac{L}{\rho^*} + O(n^\alpha)\right) + f(f+1)O(Ln^\beta)}\\
&\ge & \lim_{Q\rightarrow \infty} \left(\frac{\gamma^*+\rho^*}{\gamma^* \rho^*} + \frac{O(n^\alpha)}{L} + \frac{O(n^{\beta+2})}{Q}\right)^{-1}~~~~~~\mbox{(since $f<n/3$)}
\end{eqnarray}

Notice that for a given graph $\sg$, $\{n,\gamma^*, \rho^*, \alpha,\beta\}$ are all constants independent of $L$ and $Q$. So for sufficiently large values of $L$ and $Q$, the last two terms in  the last inequality becomes negligible compared to the first term, and the throughput of NAB approaches to a value that is at least as large as $T_{NAB}$, which is defined
\begin{equation}
T_{NAB}(\sg) = \frac{\gamma^*\rho^*}{\gamma^* +\rho^*}.
\end{equation} 

\begin{figure}[t]
\centering
\includegraphics[width = 6.5 in]{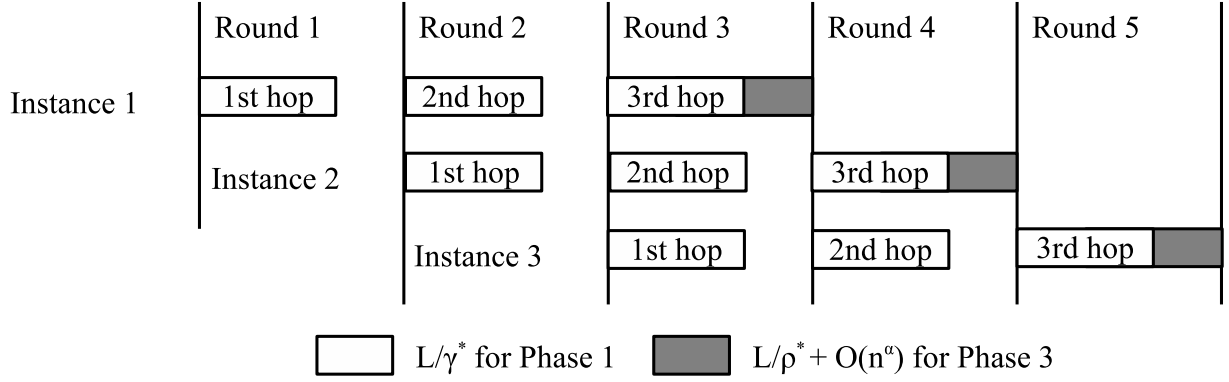}
\caption{Example of pipelining}
\label{fig:pipeline}
\end{figure}

In the above discussion, we implicitly assumed that transmissions during the unreliable broadcast in Phase 1 accomplish all at the same time, by assuming no propagation delay. However, when propagation delay is considered, a node cannot forward a message/symbol until it finishes receiving it. So for the $k$-th instance of NAB, the information broadcast by the source propagates only one hop every $L/\gamma_k$ time units. So for a large network, the ``time span'' of Phase 1 can be much larger than $L/\gamma_k$. This problem can be solved by pipelining: We divide the time horizon into rounds of $\left(\frac{L}{\gamma^*} + \frac{L}{\rho^*} + O( n^\alpha)\right)$ time units. For each instance of NAB, the $L$-bit input from the source node 1 propagates one hop per round, using the first $L/\gamma^*$ time units, until Phase 1 completes. Then the remaining $\left(\frac{L}{\rho^*} + O( n^\alpha)\right)$ time units of the last round is used to perform Phase 2. An example in which the broadcast in Phase 1 takes 3 hops is shown in Figure \ref{fig:pipeline}. By pipelining, we achieve the lower bound from Eq.\ref{eq:throughput}.
 
\section{Construction of $\Gamma$}
\label{app:subgraph}

A subgraph of $\sg$ belonging to $\Gamma$ is obtained as follows:
We will say that edges in $W \subset \se$ are ``explainable'' if there exists
a set $F \subset \sv$ such that
(i) $F$ contains at most $f$ nodes, and (ii) 
each edge in $W$ is incident on at least one node in $F$.
Set $F$ is then said to ``explain set $W$''.

Consider each {\em explainable} set of edges $W\subset \se$. Suppose that
$F_1,\cdots,F_{\Delta}$ are all the subsets of $\sv$ that {\em explain} edge set $W$.
A subgraph $\Psi_W$ of $\sg$ is obtained by removing edges in $W$ from $\se$, and nodes in
$\bigcap_{\delta=1}^{\Delta} F_{\delta}$ from $\sv$ \footnote{\normalsize It is possible that $\Psi_W$ for different $W$
may be identical. This does not affect the correctness of our algorithm.}.
In general, $\Psi_W$ above may or may not contain the source node 1. Only those
$\Psi_W$'s that do contain node 1 belongs to $\Gamma$.

\section{Proof of Theorem \ref{thm:BB:bound}}
\label{app:upperbound}

In arbitrary point-to-point  network $\sg(\sg,\se)$, the capacity of the BB problem with node 1 being the source and up to $f<n/3$ faults satisfies the following upper bounds

\subsection{$C_{BB}(\sg)\le \gamma^*$}
\label{app:upperbound:first}
\begin{proof}
Consider any $\Psi_W\in \Gamma$ and let $W$ is the set of edges in $\sg$ but not in $\Psi_W$. By the construction of $\Gamma$, there must be at least one set $F\subset \sv$ that explains $W$ and does not contain the source node $1$. 
We are going to show that $C_{BB}(\sg)\le MINCUT(\Psi_W,1,i)$ for every node $i\neq 1$ that is in $\Psi_W$.

Notice that there must exist a set of nodes that explains $W$ and does not contain node 1; otherwise node 1 is not in $\Psi_W$. Without loss of generality, assume that $F_1$ is one such set nodes.

First consider any node $i\neq 1$ in $\Psi_W$ but $i\notin F_1$. Let all the nodes in $F_1$ be faulty such that they refuse to communicate over edges in $W$, but otherwise behave correctly. In this case, since the source is fault-free, node $i$ must be able to receive the $L$-bit input that node 1 is trying to broadcast. So $C_{BB}(\sg)\le MINCUT(\Psi_W,1,i)$.

Next we consider a  node $i\neq 1$ in $\Psi_W$ and $i\in F_1$. Notice that node $i$ cannot be contained in all sets of nodes that explain $W$, otherwise node $i$ is not in $\Psi_W$. Then there are only two possibilities:
\begin{enumerate}
\item There exist a set $F$ that explaining $W$ that contains neither node 1 nor node $i$. In this case, $C_{BB}(\sg)\le MINCUT(\Psi_W,1,i)$ according to the above argument by replacing $F_1$ with $F$.
\item Otherwise, any set $F$ that explains $W$ and does not contain node $i$ must contain node 1. Let $F_2$ be one such set of nodes.

Define $V^-=\sv - F_1 - F_2$. $V^-$ is not empty since $F_1$ and $F_2$ both contain at most $f$ nodes and there are $n\ge 3f+1$ nodes in $\sv$. Consider two scenarios with the same input value $x$: (1) Nodes in $F_1$ (does not contain node 1) are faulty and refuse to communicate over edges in $W$, but otherwise behave correctly; and (2) Nodes in $F_2$ (contains node 1) are faulty and refuse to communicate over edges in $W$, but otherwise behave correctly. In both cases, nodes in $V^-$ are fault-free. 

Observe that among edges between nodes in $V^-$ and $F_1\cup F_2$, only edges between  $V^-$ and $F_1\cap F_2$ could have been removed, because otherwise $W$ cannot be explained by both $F_1$ and $F_2$. So nodes in $V^-$ cannot distinguish between the two scenarios above. In scenario (1), the source node 1 is not faulty. Hence nodes in $V^-$ must agree with the value $x$ that node 1 is trying to broadcast, according to the validity condition. Since nodes in $V^-$ cannot distinguish between the two scenarios, they must also set their outputs to $x$ in scenario (2), even though in this case the source node 1 is faulty. Then according to the agreement condition, node $i$ must agree with nodes in $V^-$ in scenario (2), which means that node $i$ also have to learn $x$. So $C_{BB}(\sg)\le MINCUT(\Psi_W,1,i)$.
\end{enumerate}
This completes the proof. 
\end{proof}

\subsection{$C_{BB}(\sg)\le 2\rho^*$}
\label{app:upperbound:second}
\newcommand{\SL}{\mathcal L}
\newcommand{\SR}{\mathcal R}

\begin{proof}
For a subgraph $H\in \Omega_1$ (and accordingly $\overline{H}\in \overline{\Omega}_1$), denote $$U_H = \min_{nodes~i,j ~in~H}MINCUT(\overline{H},i,j).$$
We will prove the upper bound by showing that $C_{BB}(G)\le U_H$ for every $H \in \Omega_1$.

Suppose on the contrary that Byzantine broadcast can be done at a rate $R > U_H + \epsilon$ for some constant $\epsilon > 0$. So there exists a BB algorithm, named $\sa$, that can broadcast $t(U_H+\epsilon)$ bits in using $t$ time units, for some $t>0$.

Let $E$ be a set of edges in $H$ that corresponds to one of the minimum-cuts in $\overline{H}$. In other words, $\sum_{e\in E} z_e = U_H$, and the nodes in $H$ can be partitioned into two non-empty sets $\SL$ and $\SR$ such that $\SL$ and $\SR$ are disconnected from each other if edges in $E$ are removed. Also denote $F$ as the set of nodes that are in $\sg$ but not in $H$. Notice that since  $H$ contains $(n-f)$ nodes, $F$ contains $f$ nodes. 

Notice that in $t$ time units, at most $t U_H < t(U_H+\epsilon)$ bits of information can be sent over edges in $E$. According to the pigeonhole principle, there must exist two different input values  of $t(U_H+\epsilon)$ bits, denoted as $u$ and $v$, such that in the absence of misbehavior, broadcasting $u$ and $v$ with algorithm $\sa$ results in the same communication pattern over edges in $E$. 

First consider the case when $F$ contains the source node 1. Consider the three scenarios using algorithm $\sa$:
\begin{enumerate}
\item Node 1 broadcasts $u$, and none of the nodes misbehaves. So all nodes should set their outputs to $u$. 
\item Node 1 broadcasts $v$, and none of the nodes misbehaves. So all nodes should set their outputs to $v$. 
\item Nodes in $F$ are faulty (includes the source node 1). The faulty nodes in $F$ behave to nodes in $\SL$ as in scenario 1, and behave to nodes in $\SR$  as in scenario 2. 
\end{enumerate}
It can be showed that nodes in $\SL$ cannot distinguish scenario 1 from scenario 3, and nodes in $\SR$ cannot distinguish scenario 2 from scenario 3. So in scenario 3, nodes in $\SL$ set their outputs to $u$ and nodes in $\SR$ set their outputs to $v$. This violates the agreement condition and contradicts with the assumption that $\sa$ solves BB at rate $U_H + \epsilon$. Hence $C_{BB}(\sg)\le U_H$.

Next consider the case when $F$ does not contain the source node 1. Without loss of generality, suppose that node 1 is in $\SL$. Consider the following three scenarios:
\begin{enumerate}
\item Node 1 broadcasts $u$, and none of the nodes misbehaves. So all nodes should set their outputs to $u$. 
\item Node 1 broadcasts $v$, and none of the nodes misbehaves. So all nodes should set their outputs to $v$. 
\item Node 1 broadcasts $u$, and nodes in $F$ are faulty. The faulty nodes in $F$ behave to nodes in $\SL$ as in scenario 1, and behave to nodes in $\SR$ as in scenario 2.
\end{enumerate}
In this case, we can also show that nodes in $\SL$ cannot distinguish scenario 1 from scenario 3, and nodes in $\SR$ cannot distinguish scenario 2 from scenario 3. So in scenario 3, nodes in $\SL$ set their outputs to $u$ and nodes in $\SR$ set their outputs to $v$. This violates the agreement condition and contradicts with the assumption that $\sa$ solves BB at rate $U_H + \epsilon$. Hence $C_{BB}(\sg)\le U_H$, and this completes the proof.
\end{proof}

\section{Proof of Theorem \ref{thm:fraction}}
\label{app:fraction}

Now we compare $T_{NAB}(\sg)$ with the upper bound of $C_{BB}(\sg)$ from Theorem \ref{thm:BB:bound}.
Recall that 
$$ T_{NAB}(\sg) ~= ~ \frac{\gamma^*\rho^*}{\gamma^*+\rho^*}$$
and
$$C_{BB}(\sg)~\le~ \min(\gamma^*,2\rho^*).$$

There are 3 cases:

\begin{enumerate}
\item $\gamma^*\le \rho^*$: Observe that $T_{NAB}(\sg)$ is an increasing function of both  $\gamma^*$ and $\rho^*$. For a given $\gamma^*$, it is minimized when $\rho^*$ is minimized. So
\begin{equation}
T_{NAB}(\sg) \ge \frac{{\gamma^*}^2}{\gamma^*+\gamma^*} = \frac{\gamma^*}{2}\ge \frac{C_{BB}(\sg)}{2}.
\end{equation}
The last inequality is due to $\gamma^*\ge C_{BB}(\sg)$.
\item $\gamma^*\le 2\rho^*$: 
\begin{eqnarray}
T_{NAB}(\sg) 
\ge  \frac{\gamma^*\rho^*}{2\rho^*+\rho^*} =\frac{\gamma^*}{3}\ge \frac{C_{BB}(\sg)}{3}.
\end{eqnarray}
The last inequality is due to $\gamma^*\ge C_{BB}(\sg)$.
\item $\gamma^*> 2\rho^*$: Since $T_{NAB}(\sg)$ is an increasing function of both  $\gamma^*$, for a given $\rho^*$, it is minimized when $\gamma^*$ is minimized. So
\begin{eqnarray}
T_{NAB}(\sg) &\ge& \frac{2{\rho^*}^2}{2\rho^*+\rho^*} = \frac{2\rho^*}{3}\ge \frac{C_{BB}(\sg)}{3}.
\end{eqnarray}
The second inequality is due to $2\rho^*\ge C_{BB}(\sg)$.
\end{enumerate}

\end{document}